\documentclass{article}

\RequirePackage{amsthm,amsmath,amsfonts,amssymb}
\RequirePackage[numbers]{natbib}

\RequirePackage[colorlinks,citecolor=blue,urlcolor=blue]{hyperref}
\RequirePackage{graphicx}

\usepackage{subfigure}
\usepackage{xcolor}
\usepackage[utf8]{inputenc}

\usepackage[a4paper,
 total={165mm,247mm},
 left=25mm,
 top=20mm]{geometry}

\newcommand{\Z}{\mathbb Z}
\newcommand{\N}{\mathbb N}

\renewcommand{\P}{\mathbb P}

\newcommand{\E}{\mathbb E}

\newcommand{\pml}{\text{\sc{PML}}}

\newcommand{\Perp}{\!\perp \! \! \! \perp\!}

\newcommand{\nei}{\text{ne}}

\newtheorem{theorem}{Theorem}[section]
\newtheorem{lemma}[theorem]{Lemma}
\newtheorem{corollary}[theorem]{Corollary}
\newtheorem{proposition}[theorem]{Proposition}
\newtheorem{definition}[theorem]{Definition}

\theoremstyle{remark}

\newtheorem{example}[theorem]{Example}

\title{Structure recovery for partially observed discrete Markov random fields on graphs under not necessarily positive distributions}

\author{Florencia Leonardi and Rodrigo Carvalho}

\begin{document}

\maketitle

\begin{abstract}
We propose a penalized pseudo-likelihood criterion to estimate the graph of 
conditional dependencies in a discrete  Markov random field that can be 
partially observed.  We prove   the convergence of the estimator  in the case 
of a finite or countable infinite set of nodes. In the finite case the underlying 
graph can be recovered with probability one, while in the countable infinite case
 we can recover any finite sub-graph with probability one, by allowing the 
 candidate neighborhoods to grow as a function $o(\log n)$, with $n$ the sample size.   Our method 
requires minimal assumptions on the probability distribution and contrary to 
other approaches in the literature, the usual positivity condition is not needed. 
 We evaluate the performance of the estimator on simulated data and we apply 
 the methodology to a real dataset of stock index markets in different countries. 
\end{abstract}

\section{Introduction}

Discrete Markov random fields on graphs,  usually called graphical models in the statistical literature,  have received much attention from researchers in recent years, especially due to their flexibility to capture conditional dependence relationships between variables \cite{lauritzen1996, koller2009,lerasle_et_al2016, pensar2017,divino2000}.
 They have been applied to many different problems in different fields such as Biology \cite{shojaie2010},  Social Sciences \cite{strauss1990} or Neuroscience \cite{duarte_et_al2019}. Graphical models are in some sense ``finite'' versions of general random fields or Gibbs distributions, classical models in stochastic processes and statistical mechanics theory \cite{georgii2011}. 

In this work we focus on discrete Markov random field  models (with a finite or  countable infinite set of variables), where the set of random variables takes values on a finite alphabet. 
One of the main statistical questions for this type of models is how to recover  the underlying graph; that is, the graph determined by the conditional dependence relationships between the variables.
 For the class of Markov random fields on lattices some methods based on penalized pseudo-likelihood criteria 
like the Bayesian Information Criterion (BIC) of  \cite{schwarz1978} 
have appeared in the literature \cite{csiszar2006b}, see also \cite{tjelmeland1998, locherbach2011}. 
On the case of Markov random fields defined on general graphs the most studied model
is  the binary graphical model with pairwise interactions where structure estimation
 can be  addressed by using standard logistic regression techniques \cite{strauss1990, 
 ravikumar2010}, distance based approaches between conditional probabilities  
 \cite{galves2015, bresler_et_al2018} and maximization of the $\ell^1$-penalized pseudo-likelihood \cite{atchade2014, hoefling09a}, see also \cite{narayana2012}.  
 In the case of bigger discrete alphabets or general type of interactions, 
 to our knowledge the only work addressing the structure estimation problem is  \cite{loh2013}, where the authors obtain a characterization of the edges in the graph with the zeros in a generalized  inverse covariance matrix. Then, this characterization is used to derive estimators for restricted classes of models and  the authors prove the consistency in probability of these estimators.
 
Markov random fields have also been proposed for continuous random variables, where the structure estimation problem has been addressed by $\ell_1$ regularization for Gaussian Markov random fields \cite{meinshausen2006} and also extended to non-parametric 
 models \cite{lafferty2012,liu2012} and  general conditional distributions from the exponential family \cite{yang2015}. 
 
All these works, for discrete or continuous random variables, assume the model satisfies a usual ``positivity''  condition, that states that the probability distributions of finite sub-sets of variables are strictly positive. The positivity condition guarantees a factorization property of the joint distribution, thanks to a classical result known as Hammersley-Clifford theorem \cite{hc1971}. But the positivity condition seem to be extremely strong for discrete distributions, where in many applications some configurations are impossible to occur and then have zero probability. Moreover, the positivity condition does not enable to consider ``sparse'' models, where many parameters are  assumed to be zero. These models are specially appealing for  high dimensional data, where the number of variables is high and the number of relevant parameters in the model must be assumed relatively small in order to have efficient estimators. 

In this work we address the structure estimation problem for discrete Markov random fields without assuming the positivity condition. We first introduce a penalized 
pseudo-likelihood criterion to estimate the neighborhood of a node, that is later combined to obtain an estimator of the underlying graph. We prove that both estimators converge almost surely  to the true underlying graph in the case of a finite graphical model when the sample size grows, without imposing additional hypothesis on the model.  
 In the countable infinite case, that is when the underlying graph is infinite and the number of observed variables is allowed to grow with the sample size, we prove that the estimator restricted to a finite sub-graph also converges almost surely to the corresponding sub-graph. 

The paper is organized  as follows. Section~\ref{sec:definitions} presents
 the definition of the model,  including some examples. 
 Section~\ref{sec:estimation} introduces the different estimators of the 
 conditional dependence graph and presents the statements proofs of the main consistency results. 
 Finally, in Section~\ref{simul} we evaluate the estimators performance through simulations, 
 in  Section~\ref{aplic} we show a real data application  
 and in the Appendix we present the proofs of some auxiliary results.

\section{Discrete Markov random fields on graphs}
\label{sec:definitions}

A \textit{graph} is a pair $G=(V,E)$, where $V$ is the set of vertices (or nodes) and $E$ is the set of edges, $E \subset V \times V$. A graph $G$ is said \textit{simple} if for all $i \in V$, $(i,i)\not \in E$ and it is said \textit{undirected} if $(i,j) \in E$ implies $(j,i) \in E$  for every pair  $(i,j) \in V \times V$. 
Given any set $S$, the symbol $|S|$ denotes its cardinality.

Let $A$ be a finite set, a \textit{random field} on $A^V$ is a family of random variables indexed by the elements of $V$, $\{X_v : v \in V\}$, where each $X_v$ is a random variable with values in $A$. For $\Delta \subseteq V$, a subset of vertices, we write $X_\Delta = \{ X_i \colon i \in \Delta \}$, and $a_\Delta = \{ a_i \in A \colon i \in \Delta \}$ denotes a configuration on $\Delta$. The law (joint distribution) of the random variables $X_V$ is denoted by $\P$.

For any finite $\Delta\subset V$ we write 
\begin{equation}\label{joint_dist}
p (a_{\Delta}) = \P (X_{\Delta}=a_{\Delta}) \; \mbox{ with }  a_\Delta \in A^\Delta
\end{equation}
and if $p (a_{\Delta}) > 0$ we denote by  
\begin{equation}\label{cond_dist}
p(a_\Phi | a_\Delta) = \P( X_\Phi= a_\Phi | X_\Delta= a_\Delta) \; \mbox{ for }  a_\Phi \in A^{\Phi}, \,a_\Delta \in A^{\Delta}
\end{equation}
the corresponding conditional probability distributions. 

Given $v\in V$, a \textit{neighborhood} $W$ of $v$ is any finite set of vertices with  $v \notin W$. 
If there is a neighborhood $W$ of $v$ 
satisfying
\begin{equation}\label{propriedadeMarkov}
p(a_v | a_W) \; =\;  p(a_v|a_\Delta)
\end{equation}
for all finite $\Delta\supset W,\,  v \notin \Delta$ and all $a_v\in A, a_\Delta \in A^\Delta$ with 
$p(a_\Delta)>0$, then $W$ is called \emph{Markov neighborhood} of $v$. The definition of a 
Markov neighborhood $W$ is equivalent to request that for all $\Phi$ finite (not containing $v$) 
with $\Phi\cap W = \emptyset$,   $X_\Phi$  is  conditionally independent of $X_v$, 
given $X_W$. Formally, 
\begin{equation}\label{local_markov}
X_v \Perp X_\Phi | X_W \,, \, \mbox{ for all }\Phi \mbox{ with }\Phi \cap W=\emptyset,
\end{equation}
where $\Perp$ is the usual symbol denoting independence of random variables. This conditionally independence assumption defining the Markov neighborhoods corresponds to the property known as \emph{local Markov}
in finite graphical models,  that is weaker than the usually assumed \emph{global Markov} property,  see \cite{lauritzen1996} for details. 

A basic fact that we can derive from the definition is that if $W$ is a Markov 
neighborhood of $v\in V$,  then any  finite set $\Delta\supset W$ is  also a Markov neighborhood of $v$.  
On the  other hand,  if $W_1$ and $W_2$ are Markov neighborhoods of $v$ then it is not always true in general that $W_1\cap W_2$  is a Markov neighborhood, as shown in the following example. 

\begin{example}
Let  $V=\{1,2,3\}$ and consider the vector $(X_1,X_2,X_3)$ of Bernoulli random variables with
 $\P(X_i=0)=1/2$,   $\P(X_i=1)=1/2$, for $i=1,2,3$. Suppose that  $X_1=X_2=X_3$ 
 with probability 1. Then it is easy to check that both $\{2\}$ and $\{3\}$ are Markov
  neighborhoods of node $1$, but the intersection is not a Markov neighborhood
   (which will imply $X_1$ being independent of $X_2$ and $X_3$). 
\end{example}

This property satisfied by some probability measures takes us to define the following Markov intersection condition, formally stated here. \\

\noindent{\bf  Markov intersection property:} For all $v\in V$ and all $W_1$ and $W_2$ Markov neighborhoods of $v$, the set $W_1\cap W_2$ is also  a Markov neighborhood of $v$.\\

The Markov intersection condition  is desirable in this context to define the smallest Markov neighborhood of a node and to enable the  structure estimation problem to be well defined. 
This property is guaranteed under the following usual condition assumed in the literature:\\

\noindent{\bf  Positivity condition:} For all finite  $W\subset V$ and all $a_W\in A^W$ we have  $p(a_W)>0$.  \\

It can be seen that the positivity condition implies the Markov intersection property, 
see \cite{lauritzen1996} for details. 
For this reason, in the literature of Markov random fields it is generally assumed that the positivity condition holds. 
But  
there are distributions satisfying the Markov intersection property that are not strictly positive. An example of this is a  typical realization of a Markov chain with some zeros in the transition matrix.
The following basic example shows that positivity condition and Markov intersection property are not equivalent in general. \\

 \begin{example}\label{ex-markov-chain}
Let  $V=\Z$ and take a stationary Markov chain of order one assuming values in $A=\{0,1\}$ with transition matrix 
\[
P = \begin{pmatrix}
1/2 & 1/2 \\
1 & 0
\end{pmatrix}
\]
Evidently, the distribution $\P$ of the Markov chain does not satisfy the positivity condition 
(any configuration with two subsequent one's has zero probability).  But the distribution satisfies the Markov intersection property, because any Markov neighborhood of  node $i$ necessarily contains nodes $i-1$ and $i+1$, that corresponds to the \emph{minimal} Markov neighborhood of node $i$. \\
\end{example}

From now on we assume the distribution $\P$ satisfies  the Markov intersection property  defined before. 
For  $v \in V$, let $\Theta(v)$ be the set of all subsets of $V$ that are Markov neighborhoods of $v$.
The \emph{basic neighborhood} of $v$  is defined as
\begin{equation}\label{ne(v)}
\nei(v) = \bigcap_{W \in \Theta(v)} W .
\end{equation}
By the  Markov intersection property, $\nei(v)$ is the smallest Markov neighborhood of $v \in V$.
Based on these special neighborhoods, define the graph $G=(V,E)$
 by 
\begin{equation}\label{graph}
(v,w) \in E \; \mbox{if and only if } w \in \nei(v). 
\end{equation}

We can state the following  basic result. 

\begin{lemma}\label{undirect}
The graph $G$, defined by \eqref{graph}, is undirected, i.e. if $(v,w) \in E \Rightarrow (w,v) \in E$.
\end{lemma}

The proof of Lemma~\ref{undirect} can be found in the Appendix. \\

As an illustration,  we show in Figure~\ref{graphs}  different Markov random field models 
with finite as well as infinite undirected graphs. 
Besides the graphs with infinite set of nodes in this example are regular in the sense that the neighborhoods 
have the same structure for each node, in our setting we allow for different neighborhood structures for different nodes, 
not imposing a model over a regular lattice as the approach considered in  \cite{csiszar2006b}. As an example, we present a joint distribution on five nodes 
having the graph in Figure~\ref{graphs}(c) as the graph of conditional dependencies between nodes.  

\begin{figure}[t!]
\begin{center}
\begin{minipage}{6cm}
\begin{center}
\includegraphics*{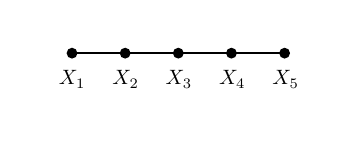}
(a)\\[14mm]
\includegraphics*[scale=0.9]{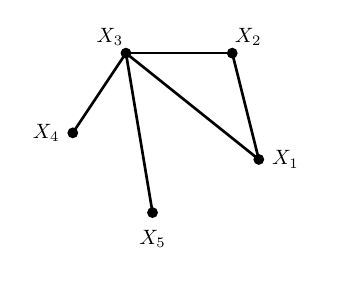}\\
(c)
\end{center}
\end{minipage}
\begin{minipage}{6cm}
\begin{center}
\includegraphics*{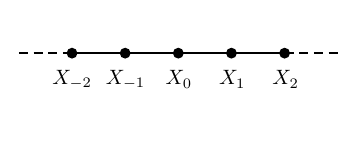}\\
(b)\\[5mm]
\includegraphics*[scale=0.7]{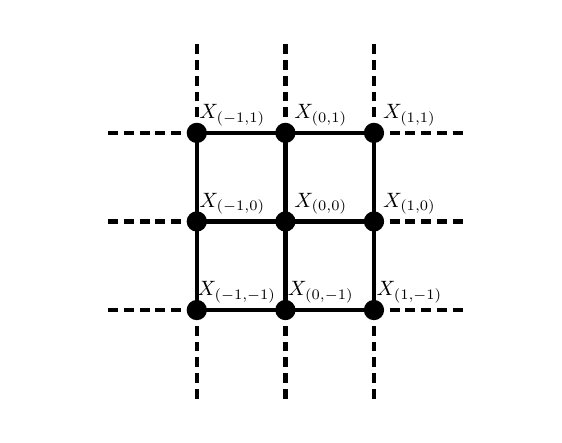}\\[3mm]
(d)\\[3mm]
\end{center}
\end{minipage}
\end{center}
\caption{Different graph structures of Markov random fields, with finite (left) and countable infinite (right) set of variables.
Examples (a) and (b) are obtained in particular 
by the distribution in Example~\ref{ex-markov-chain}. Case (a) is the projection on $\{0,1\}^{\{1,\dotsc,5\}}$
and case (b) is a representation of the joint distribution on $\{0,1\}^{\mathbb Z}$. 
Figure~\ref{graphs}(c) is a finite graphical model defined on a general graph that we will 
further use  in the simulations in Section~\ref{simul}, see Example~\ref{ex-graph} below  and Figure~\ref{graphs}(d) represents 
the interaction graph in a classical Ising model, see for example \citep{csiszar2006b,georgii2011}. }
\label{graphs}
\end{figure}

\begin{example}\label{ex-graph}
Consider the alphabet $A=\{0,1,2\}$ and define the joint probability distribution of the vector $(X_1,X_2,\dots, X_5)$
using the factorization 
\begin{equation}\label{joint-dist-c}
p(x_1,x_2,x_3,x_4,x_5) \;= \;  p(x_3)p(x_2|x_1,x_3)p(x_1|x_3)p(x_4|x_3)p(x_5|x_3)\,
\end{equation}
with the conditional distributions given in Table~\ref{table-probs}. This distribution does not satisfy the positivity condition,
 but the basic neighborhoods are well defined for each node and its graph of conditional dependencies is given by Figure~\ref{graphs}(c). We consider this particular distribution later in the 
simulations in Section~\ref{simul}. 
\end{example}

\begin{table}[t!]
\begin{center}
\begin{minipage}{7cm}
\begin{center}
\begin{tabular}{c|ccc}
$x_3$ & 0 & 1 & 2 \\ 
\hline 
$p(x_3)$ & $0.3$ & $0.2$ & $0.5$ \\ 
\end{tabular}

\vspace*{6mm}

\begin{tabular}{c|ccc}
$x_2$ & 0 & 1 & 2 \\ 
\hline 
$p(x_2|x_1=0,x_3=0)$ & $0.5$ & $0.5$ & $0$ \\ 
$p(x_2|x_1=1,x_3=0)$ & $0.5$ & $0.25$ & $0.25$ \\ 
$p(x_2|x_1=2,x_3=0)$ & $0.25$ & $0.25$ & $0.5$ \\ 
$p(x_2|x_1=0,x_3=1)$ & $0.3$ & $0$ & $0.7$ \\ 
$p(x_2|x_1=1,x_3=1)$ & $0.25$ & $0.25$ & $0.5$ \\ 
$p(x_2|x_1=2,x_3=1)$ & $0.3$ & $0.7$ & $0$ \\
$p(x_2|x_1=0,x_3=2)$ & $0$ & $0.75$ & $0.25$ \\ 
$p(x_2|x_1=1,x_3=2)$ & $0.3$ & $0.3$ & $0.4$ \\ 
$p(x_2|x_1=2,x_3=2)$ & $0.4$ & $0.3$ & $0.3$ \\
\end{tabular}
\end{center}
\end{minipage}
\begin{minipage}{7cm}
\begin{center}
\begin{tabular}{c|ccc}
$x_1$ & 0 & 1 & 2 \\ 
\hline 
$p(x_1|x_3=0)$ & $0.2$ & $0.4$ & $0.4$ \\ 
$p(x_1|x_3=1)$ & $0.3$ & $0.4$ & $0.3$ \\ 
$p(x_1|x_3=2)$ & $0.4$ & $0.3$ & $0.3$ \\ 
\end{tabular}
\vspace*{4mm}

\begin{tabular}{c|ccc}
$x_4$ & 0 & 1 & 2 \\ 
\hline 
$p(x_4|x_3=0)$ & $0.1$ & $0.4$ & $0.5$ \\ 
$p(x_4|x_3=1)$ & $0.2$ & $0.7$ & $0.1$ \\ 
$p(x_4|x_3=2)$ & $0.3$ & $0.6$ & $0.1$ \\ 
\end{tabular}
\vspace*{4mm}

\begin{tabular}{c|ccc}
$x_5$ & 0 & 1 & 2 \\ 
\hline 
$p(x_5|x_3=0)$ & $0.2$ & $0.6$ & $0.2$ \\ 
$p(x_5|x_3=1)$ & $0.3$ & $0.1$ & $0.6$ \\ 
$p(x_5|x_3=2)$ & $0.4$ & $0.3$ & $0.3$ \\ 
\end{tabular}
\end{center}
\end{minipage}
\end{center}
\caption{Conditional probabilities used to define a joint distribution on $\{0,1,2\}^5$ by the factorization 
$p(x_1,x_2,x_3,x_4,x_5)= p(x_3) p(x_2|x_1,x_3)p(x_1|x_3)p(x_4|x_3)p(x_5|x_3)$. As some conditional probabilities
are 0 the joint distribution does not satisfy the positivity condition. The graph of conditional dependencies of the vector 
$(X_1,X_2,\dots,X_5)$ is given by Figure~\ref{graphs}(c).}
\label{table-probs}
\end{table}

\section{Estimation and model selection}\label{sec:estimation}

Suppose we (partially) observe an independent sample with size $n$ of the random field $\{X_v\colon v\in V\}$ with distribution $\P$. 
Let $\{V_n\}_{n\in\N}$ be a sequence of finite subsets of $V$ and denote by $x^{(i)}_{v}$ the value obtained at the vertex $v$ 
on the $i$-th observation of the sample. When $V$ is finite we assume, without loss  of generality, that $V_n=V$ for all $n\in\N$.  
If $V$ is infinite, we assume 
 $V_n\nearrow V$ when $n\to\infty$; then, as $\nei(v)$ is assumed to be  finite,  $V_n$ will eventually contain the set $\{v\}\cup\{\nei(v)\}$.
The structure estimation problem  consists on determining the set of neighbors $\nei(v)$ for some target nodes $v$ belonging to a finite set, 
based on the partial  sample $\{x_v^{(1:n)}\colon v\in V_n\}$. Recovering the neighbors of a set of nodes enables us to recover the induced sub-graph of $G$ over this set, as we  prove in Corollary~\ref{cor:main2}. 

Given a vertex $v \in V_n$ and a set $W\subset V_n$ not containing $v$, the operator $N(a_v,a_W)$ will denote the number of occurrences of the event
\[
\{X_v=a_v\}\cap \{X_W=a_W\}
\]
in the sample. That is
\[
N(a_v,a_W) \;=\; \sum_{i=1}^n\mathbf{1}\{x_v^{(i)}=a_v\,,\, x_W^{(i)}=a_W\}\,.
\]
The conditional likelihood function of $X_v$ given $X_W=a_W$, for a set of parameters $\{q_a\colon a\in A\}$, is then 
\begin{equation}\label{cond_like}
L( (q_a)_{a\in A} | x_W^{(1:n)}) =  \prod_{a\in A}  q_{a}^{N(a,a_W)} 
\end{equation}
and it is not hard to prove that the distribution over $A$ maximizing this 
function is given by
\begin{equation}\label{trans}
q_a = \hat p(a_v|a_W) = \frac{N(a_v,a_W)}{N(a_W)},\qquad a\in A\,,
\end{equation}
for  all $a_W$ with $N(a_W)>0$, where $N(a_W) = \sum_{a_v\in A} N(a_v,a_W)$.
By multiplying all the maximum likelihoods of the conditional distribution of $X_v$ given $X_W=a_W$ for the different $a_W\in A^W$ we can compute a maximal 
pseudo-likelihood function for vertex $v\in V$, given  by   
\begin{equation}\label{hatp}
\hat\P( x_v^{(1:n)} | x_{W}^{(1:n)}) = \prod_{a_W\in A^{W}}\prod_{a_v\in A} \hat p(a_v|a_W)^{N(a,a_W)} \,,
\end{equation}
where the product is over all $a_W\in A^W$ with $N(a_W)>0$ and all $a_v\in A$ with $\hat p(a_v|a_W)>0$.  

Before presenting the main definitions and results of this section we state a proposition 
of independent interest, that shows a non asymptotic upper bound for 
the rate of convergence of $\hat{p}(a_v |a_W)$ to $p(a_v |a_W)$. This proposition
is  related to  a result obtained in \cite{garivier2011} for the estimation of the context tree of a stationary and ergodic process and its proof is 
given in the appendix.

\begin{proposition}\label{prop:bound}
For all $\delta>0$, $n\geq \exp(\delta^{-1})$, $v\in V_n$, $W \subset V_n\setminus\{v\}$ and all $a_W \in A^{W}$  we have 
\[
\P\Bigl( N(a_W) \,\sup_{a_v\in A} | \hat{p}(a_v |a_W) -  p(a_v |a_W)|^2 \;>\; \delta \log n \Bigr)\;\leq\; \frac{2|A|\delta\log^2 n}{n^\delta}\,.
\]
\end{proposition}

We are now ready to introduce the following neighborhood estimator  for the set $\nei(v)$, for $v\in V$.

\begin{definition}\label{def:est_bounded}
Given a partial sample $\{x_v^{(1:n)}\colon v\in V_n\}$ and a constant $c>0$, the empirical neighborhood of $v\in V_n$ is the set of vertices 
$\widehat \nei(v)$ defined by
\begin{equation}\label{argmax}
\widehat \nei(v) \;=\; \underset{W\subset V_n\setminus\{v\}}{\arg\max}\bigl\{ \,\log \hat\P(x_v^{(1:n)}| x_{W}^{(1:n)})  - c\,|A|^{|W|}\log n \,\bigr\}\,.
\end{equation}
\end{definition}

In order to state our main results, we recall the definition of  the  K\"ullback-Leibler divergence between two probability distributions $p$ 
and $q$ over $A$. It is given by
 \begin{equation}\label{KL}
 D(p;q) = \sum_{a\in A} p(a)\log\frac{p(a)}{q(a)}
 \end{equation}
where, by convention, $p(a) \log\frac{p(a)}{q(a)}=0$ if $p(a)=0$ and $p(a) \log\frac{p(a)}{q(a)}=+\infty$ if $p(a) > q(a)=0$. An important property of the K\"ullback-Leibler divergence is that  $D(p;q) =0$ if and only if $p(a)=q(a)$ for all $a\in A$.

For any $v\in V$ denote by  
\[
p_n(v) = \min_{\substack{a_v,a_{W}\\W\subset V_n\setminus\{v\}}}
 \{\;p(a_v|a_{W})\colon p(a_v|a_{W})>0\;\}
\]
and 
\[
\alpha_n(v)  =  \min_{\substack{W\subset V_n\setminus\{v\}\\
\nei(v)\not\subset W}}\;\Bigl\{  \sum_{a_{W\cup\nei(v)}} \;  p(a_{W\cup \nei(v)}) D ( p(\cdot_v|a_{\nei(v)})\,;\, p(\cdot_v|a_{W})) \Bigr\}
\]
where $p(\cdot_v|a_{\nei(v)})$ denotes the probability distribution over $A$ given by  $\{p(a_v|a_{\nei(v)}\}_{a_v\in A}$ and similarly for $p(\cdot_v|a_W)$. 
We note that for  any vertex $v\in V$ and any $n\in\mathbb N$ we must have $p_n(v)>0$ and  also by the definition of the basic neigborhood and Lemma~\ref{vizinhanca_contida1} we must have  $\alpha_n(v)>0$. For simplicity in the proofs and as is usual in the literature, we will assume that this quantities are uniformly bounded from below by positive constants, that is we assume that
\[
p_*\;=\; \inf_{v}\inf_{n} \{\; p_n(v)\;\} \quad \text{ and }  \quad \alpha_* \;=\;\inf_v \inf_{n} \{\; p_n(v)\;\}
\]
are positive constants. Observe that $p_*>0$ is not equivalent to the positivity condition, where \emph{all} the conditional probabilities are assumed to be positive. Here we are assuming that those positive probabilities in the model are bounded from below, but some of them can be zero. 
%

We can now state the following consistency result for the neighborhood estimator given in \eqref{argmax}.

\begin{theorem}\label{main}
Let $v\in V$.  Assume $V_n\nearrow V$ with $|V_n| = o(\log n)$. Assume also that $p_*>0$ and
$\alpha_*>0$.  Then for any $c>0$, the estimator given by \eqref{argmax}  satisfies 
$\widehat\nei(v) = \nei(v)$ with probability converging to 1 as $n\to \infty$. Moreover, if $c > 
 |A|^2[p_*(|A|-1)]^{-1}$ then $\widehat\nei(v) = \nei(v)$ eventually almost surely as $n\to \infty$.
\end{theorem}

Once we have guaranties that we can consistently estimate the neighborhood of a node $v\in V$, we can consider the estimation of a finite subgraph of $G$. To do this, we  can simply  estimate the neighborhood of each node and reconstruct the subgraph based on the set of estimated  neighborhoods. Given a set $V'\subset V$, we denote by $G_{V'}$ the induced subgraph; that is, the graph given by the pair $(V', E')$, where $E' = \{(v,w)\in E\colon v,w\in V'\}$.
Based on the neighborhood estimator (\ref{argmax}), we can construct an estimator of the subgraph $G_{V'}$  by defining the set of edges
\begin{equation}\label{conserv-graph}
\widehat E'_\wedge = \{(v,w)\in V'\times V'\colon v\in \widehat\nei(w) \text{ and }w\in \widehat\nei(v)\}
\end{equation}
were this is refereed as the \emph{conservative approach} or we can take 
\begin{equation}\label{nconserv-graph}
\widehat E'_\vee = \{(v,w)\in V'\times V'\colon v\in \widehat\nei(w) \text{ or }w\in \widehat\nei(v)\}\,,
\end{equation}
a \emph{non-conservative} approach. \\

Theorem~\ref{main} then implies the following strong consistency result for any $G'$ 
with a finite set of vertices $V'$.

\begin{corollary}\label{cor:main2}
Let $G'= (V', E')$ be  an induced sub-graph of $G$ with a finite set of vertices $V'$, and assume
the hypotheses of Theorem~\ref{main} hold.  
Then, for  $c > 0$ (respectively $c>  |A|^2[p_*(|A|-1)]^{-1}$),  if  $|V_n| = o(\log n)$
we have $\widehat E'_{\wedge} = \widehat E'_{\vee} =  E'$ with probability
converging to 1  as $n\to \infty$ (respectively eventually almost surely as $n\to \infty$). 
\end{corollary}

The proofs of all the theoretical results in this section, as well as some auxiliary results, are presented in the Appendix. 

\section{Simulations}\label{simul}

In this section we show the results of a simulation study to evaluate the 
performance of the graph estimators  \eqref{conserv-graph} and  
\eqref{nconserv-graph} on different sample sizes and for different values of
 the penalizing constant $c$.

We simulated a probability distribution on five vertices with alphabet $A=\{0,1,2\}$
and with graph of conditional dependencies given by Figure~\ref{graphs}(c). 
The joint distribution is assumed to factorize as in \eqref{joint-dist-c},
having conditional probabilities given by Table~\ref{table-probs}.

\begin{figure}[t]
\begin{center}
\includegraphics*[scale=0.62]{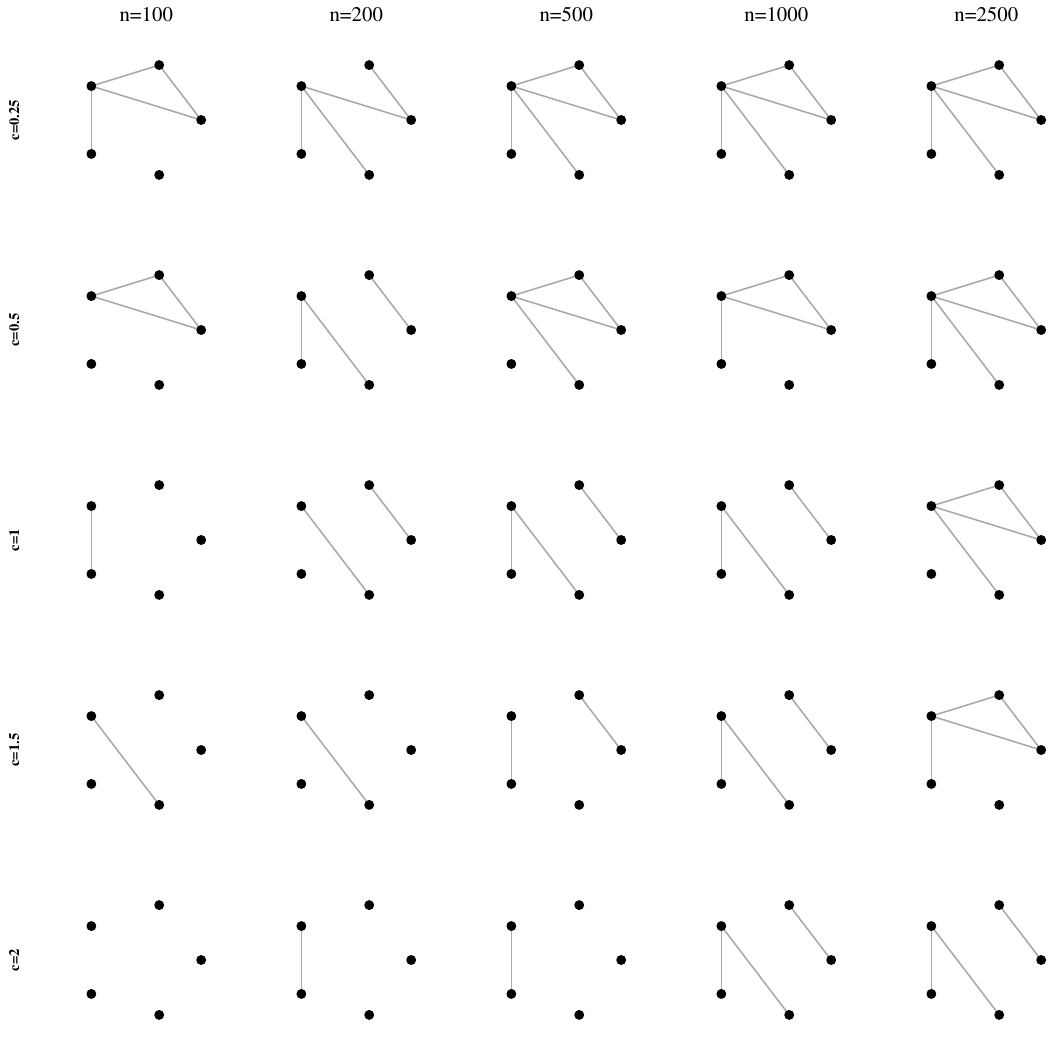}
   \caption{Estimated graph for
    neighbohood approach.}\label{sim-graph-cons}
 \end{center}
\end{figure}

  \begin{figure}[t]
\begin{center}
 \includegraphics*[scale=0.62]{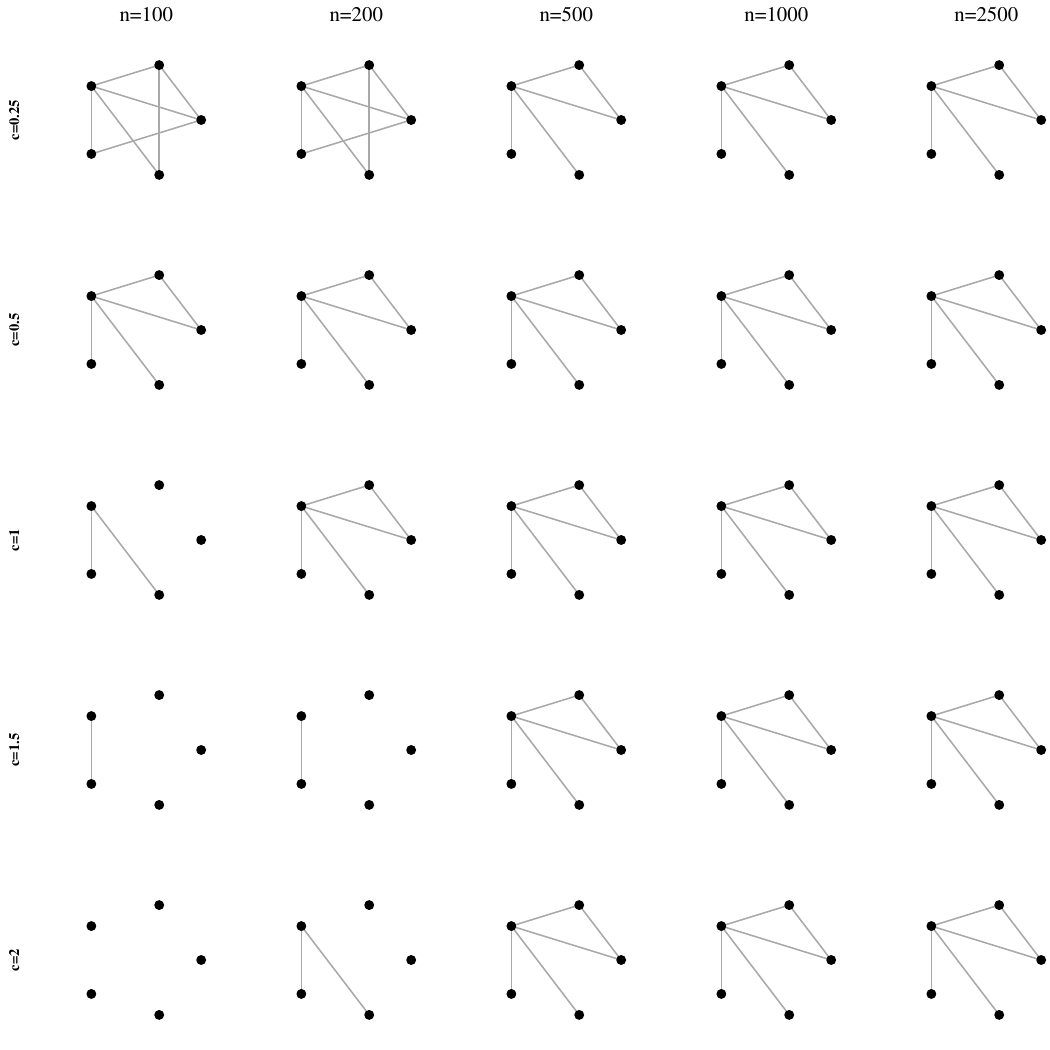}
   \caption{Estimated graph for
    neighbohood approach.}\label{sim-graph-ncons}
\end{center}
\end{figure}

The algorithms implementing the  graph estimators given by  
\eqref{conserv-graph}  (conservative approach) and  \eqref{nconserv-graph} (non-conservative approach) were coded in the {\tt R}
language and are available as a package called {\tt mrfse} in CRAN.
In Figures~\ref{sim-graph-cons} and \ref{sim-graph-ncons}
 we show the results of both approaches, for values of the penalising constant  $c$ in the set $\{0.25,0.5,1,1.5,2\}$
and sample sizes $n$ in the set $\{100,200,500,1000,2500\}$. 
In this example the non-conservative approach seems to converge to the 
underlying graph faster  than the conservative approach. 
To evaluate this difference in a quantitative form, we  compute a numeric value for
 the underestimation error (\emph{ue}), overestimation error (\emph{oe}) and total error (\emph{te}), 
 given by 
\begin{equation}\label{ue}
ue\;=\; \dfrac{\sum_{(v,w)}\mathbf{1}\{(v, w) \in E
  \;\text{and}\; (v, w) \not\in \widehat{E})\}}{\sum_{(v,w)}\mathbf{1}\{(v, w) \in E\}}
\end{equation}
\begin{equation}\label{oe}
oe\;=\; \dfrac{\sum_{(v,w)}\mathbf{1}\{(v, w) \not\in E
  \;\text{and}\; (v, w) \in \widehat{E})\}}{\sum_{(v,w)}\mathbf{1}\{(v, w) \notin E\}}\;
\end{equation}
and
\begin{equation}\label{te}
te\;=\; \dfrac{ oe  \sum_{(v,w)} \mathbf{1}\{(v, w) \in
  E \}  + ue \sum_{(v,w)} \mathbf{1}\{(v, w) \not\in E\} }{|V|(|V|-1)}\,.
\end{equation}
Figure~\ref{graph1-errors} shows an evaluation of \emph{ue}, \emph{oe} 
and \emph{te} for both methods,  used with  constant c = 1, and
with sample sizes ranging from 100 to 10,000. 
The error lines corresponds to the mean errors in 30 runs of both algorithms.  

\begin{figure}[t]
\begin{center}
\includegraphics*[scale=0.47]{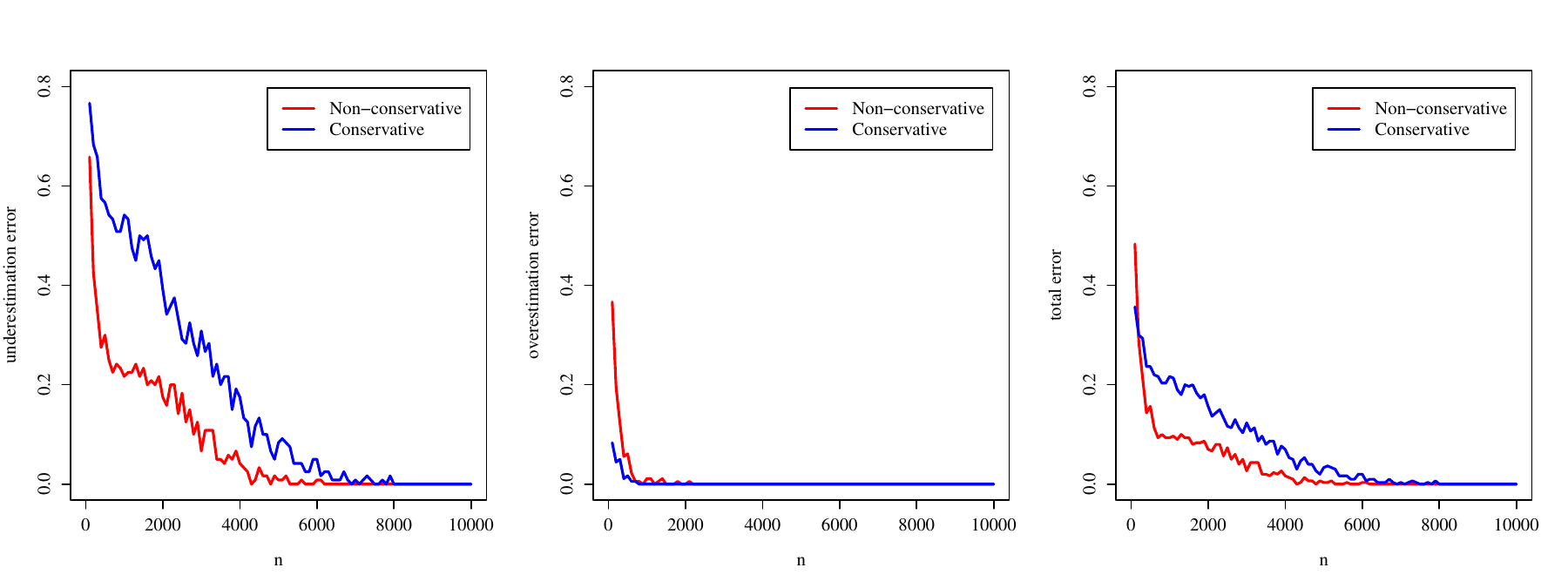}
   \caption{Mean of underestimation error (left), overestimation error (center)
   and total error (right) defined by \eqref{ue}-\eqref{te}, computed on 30 runs of the simulation, for both conservative and non-conservative approaches.}\label{graph1-errors}
 \end{center}
\end{figure}

We also performed simulations for other distributions with different  
graph structures, specifically for the projected Markov chain in five variables 
given by Figure~\ref{graphs}(a) (a graph with few edges) and the complete graph
with five vertices (a graph with a maximum number of edges). 
We also simulated a distribution with 15 nodes, the same number of nodes considered in the application
in Section~\ref{aplic}. In all cases, 
the non-conservative approach seems to perform better than the conservative algorithm. 
The results are reported in the Supplementary material to this article.

\section{Application on real data}\label{aplic}

To illustrate the performance of the estimator on real data we
analyzed a stock index from fifteen countries on different times taken
from the site
\url{https://br.investing.com/indices/world-indices}. The countries
are Brazil, USA, UK, France, India, Japan, Greece, Ireland, South
Africa, Spain, Marroco, Australia, Mexico, China and Saudi Arabia with
stock markets Bovespa, NASDAQ, FTSE 100, CAC 40, Nifty 50, Nikkei 225,
FTSE ATHEX Large Cap, FTSE Ireland, FTSE South Africa, IBEX 35,
Moroccan All Shares, S\&P ASX 200, S\&P BMV IPC, Shanghai Composite and
Tadawul All Share, respectively. We collected 2120 entries where each
entry contains the indicator function of an increasing variation in
the stock index for a given day with respect to the previous day, for
each one of the fifteen stock markets. That is, a stock market for a given
day $d$ is codified as 1 if the stock index at day $d$ is greater than
the stock index at day $d-1$, and 0 otherwise.  The main goal is to
estimate the conditional dependence graph between the codified stock
markets corresponding to the fifteen countries.
The dataset of stock index variation corresponds to subsequent time
points (days) in the period from December $29^{th}$ of 2010 to October
$22^{th}$ of 2018. To reduce sample correlation between subsequent
observations we selected data points with a difference of 4 days.  The
final sample has a total of 530 time points. The datas are available
in \url{https://github.com/rodrigorsdc/ic/tree/master/stock_data}

We applied the two approaches given by \eqref{conserv-graph} and  
\eqref{nconserv-graph} to estimate the conditional dependence graph. We
chose the penalising constant as $c = 0.2$ by 10-Fold Cross-validation.
The resulting graphs with the conservative and non-conservative
approaches are shown in Figure \ref{fig:stock_map_neigh}. In both
cases, the obtained graphs connect countries that are geographically
near, as could be somehow expected. Also, the conservative approach
underestimates a few edges compared to the non-conservative.

\begin{figure}[t!]
\begin{center}
 \includegraphics*[scale=0.25,trim=0 0 0 0,clip=true]{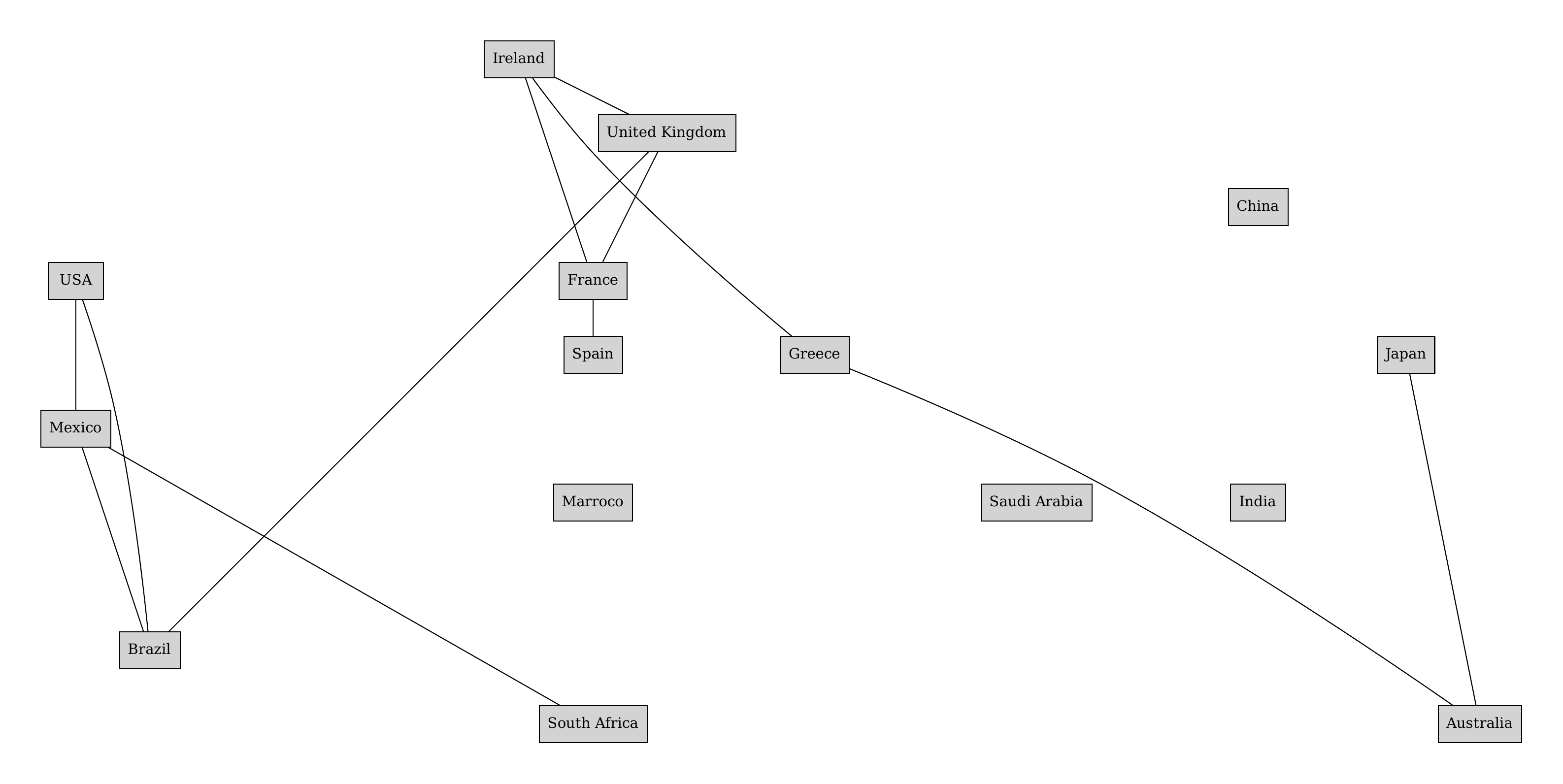}\\[5mm]
\includegraphics*[scale=0.25,trim=0 0 0 0,clip=true]{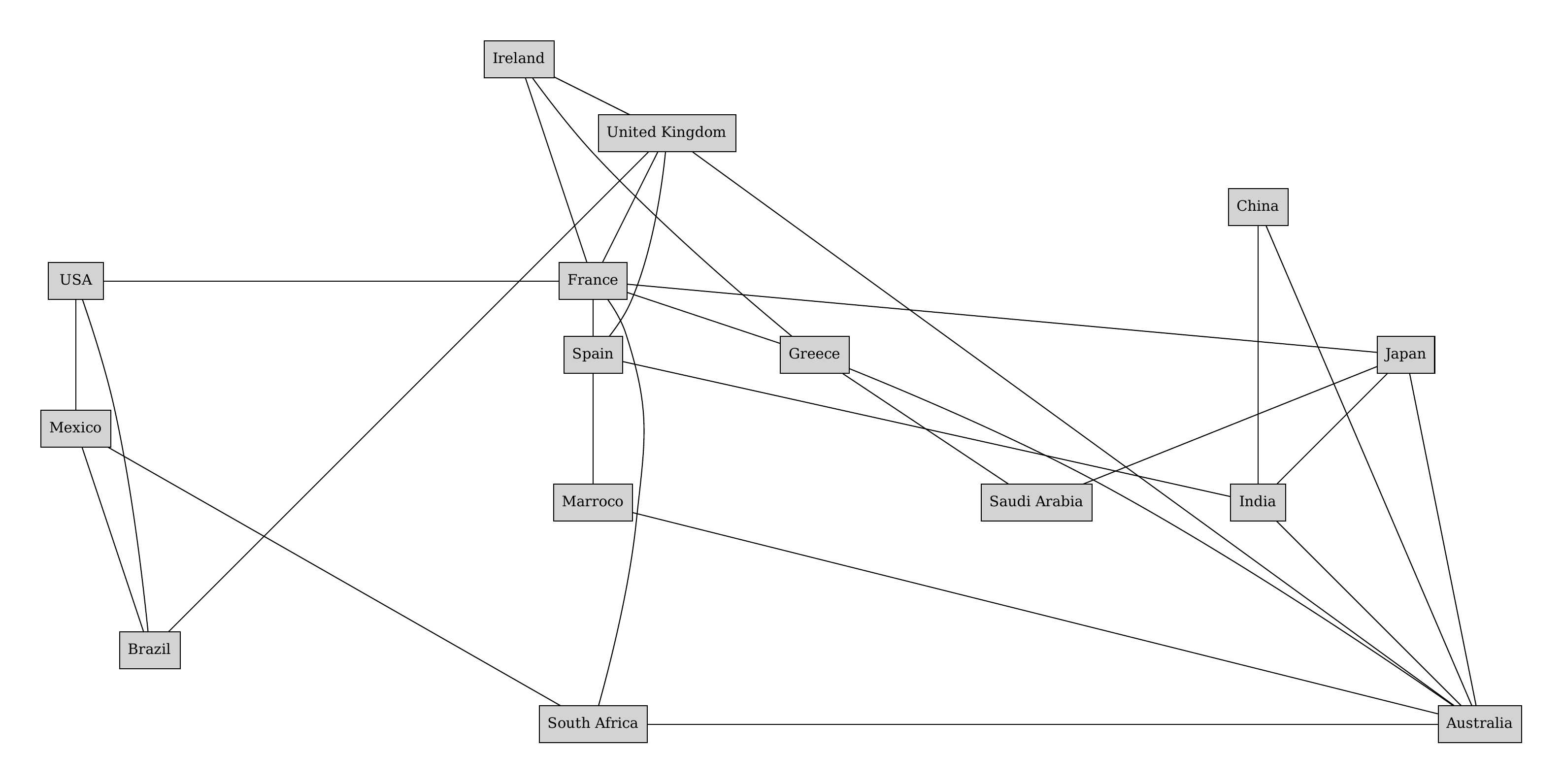}
\end{center}
   \caption{Estimated graphs with the conservative (top) and non-conservative (bottom) approaches. The constant $c$ in \eqref{def:est_bounded} 
     used in both cases was $0.2$.}\label{fig:stock_map_neigh}
\end{figure}

\section*{Discussion}

In this paper we introduced an estimator for the basic neighborhood of a node 
in a general discrete Markov  random field defined on a graph. We  showed that the estimator is consistent for any value of the penalizing constant and is strongly consistent for a sufficiently large value of the constant. This result implies that any finite sub-graph can be recovered with probability one when the sample size diverges, provided 
the set of  observed nodes increases not too fast or contains all the target nodes and its neighborhoods. The proof is based on some  
deviation inequalities given in Proposition~\ref{prop:bound}, a result that is 
derived from a martingale approach appearing in \cite{garivier2011} but that is new in this context of Markov random fields. 
One advantage of our results is that we do not need to assume a positivity condition, namely that all conditional probabilities in the model are strictly 
positive. This allows us to consider  sparse models, that is models that can have many parameters equal to zero and then a low number of significant  parameters, a property that is  appealing on high dimensional contexts. 
 We consider the samples of the Markov random field are independent and 
 identically distributed, but one important question to address in future work is if 
 this method can be generalised to dependent data, as for example the case of mixing processes as considered in  \cite{leonardi2020}. 

\section*{Acknowledgemnts}

This article was produced as part of the activities of FAPESP  Research, 
Innovation and Dissemination Center for Neuromathematics, grant 
2013/07699-0, and FAPESP's project ``Model selection in high dimensions: theoretical properties and applications'',  grant 2019/17734-3, S\~ao Paulo Research 
Foundation. F.L is partially supported by a fellowship from ``Conselho Nacional de Desenvolvimento Científico e Tecnológico -- CNPq'', grant 311763/2020-0.

\section*{Appendix: proof of theoretical results}

Here we present the proofs of  the main results in the paper, namely Lemma~\ref{undirect}, Proposition~\ref{prop:bound}, Theorem~\ref{main} and Corollary~\ref{cor:main2}.
We also prove some auxiliary results needed to  demonstrate the
main theorem in the article. 

\begin{proof}[Proof of Lemma~\ref{undirect}]
Suppose  $w \notin \nei(v)$. Take any $\Delta$  such that 
$\nei(v)\subseteq\Delta, v\notin\Delta$.  Then
\begin{equation}\label{nao-vizinho}
p(a_v | a_{\Delta})=p(a_v | a_{\Delta\setminus \{w\}})\quad \mbox{ for  all }a_{\Delta} \in A^\Delta\text{ with } p(a_\Delta)>0.
\end{equation}
By the definition of conditional probability and (\ref{nao-vizinho}) we have that
\begin{align}\label{prob-igualdade-1}
p(a_v,a_w | a_{\Delta\setminus \{w\}})&=p(a_v|a_{\Delta},a_w) p(a_w|a_{\Delta\setminus\{w\}}) \\
&= p(a_v|a_{\Delta\setminus \{w\}}) p(a_w|a_{\Delta\setminus \{w\}})
\notag
\end{align}
for  all $a_\Delta$ with $p(a_\Delta)>0$. 
But we also have that 
\begin{equation}\label{prob-igualdade-2}
p(a_v,a_w | a_{\Delta\setminus \{w\}})=p(a_w|a_{\Delta\setminus \{w\}},a_v) p(a_v|a_{\Delta\setminus \{w\}}). 
\end{equation}
Therefore, by (\ref{prob-igualdade-1}) and (\ref{prob-igualdade-2}), if  $p(a_v|a_{\Delta\setminus \{w\}}) > 0$ we obtain 
\begin{equation*}
p(a_w|a_{\Delta\setminus \{w\}},a_v)= p(a_w|a_{\Delta\setminus \{w\}})\,.
\end{equation*}
As the equality holds for all $\Delta$ and all  $(a_v, a_{\Delta\setminus \{w\}})$ 
with $p(a_v, a_{\Delta\setminus \{w\}}) >0$ then we conclude that $v\notin\nei(w)$.
\end{proof}

\begin{proof}[Proof of Proposition~\ref{prop:bound}]
First observe that 
\begin{equation}\label{first_eq}
\begin{split}
\P\Bigl( N(a_W) \,\sup_{a_v\in A}& | \hat{p}(a_v |a_W) -  p(a_v |a_W)|^2 \;>\; \delta \log n \Bigr)\\
&\;\leq\; \sum_{a_v\in A} \P\Bigl( N(a_W) \, | \hat{p}(a_v |a_W) -  p(a_v |a_W)|^2 \;>\; \delta \log n \Bigr)
\end{split}
\end{equation}
then we will fix $a_v\in A$ and bound above each term in the right hand 
side separately. For simplifying the notation we write $\hat{p}_n= \hat{p}(a_v |a_W)$, $p= p(a_v |a_W)$, $O_n=N(a_v,a_W)$ and $N_n=N(a_W)$. Observe that 
$\hat{p}_n = O_n/N_n$. 
For $\lambda>0$ define $\phi(\lambda) = \log(1- p+2^{\lambda}p)$. 
Let $W^{\lambda}_0=1$ and for $n\geq 1$ define
\[
W^{\lambda}_n = 2^{\lambda O_n- N_n\phi(\lambda)}\,.
\]
Observe that $W^{\lambda}_n$ is a martingale with respect to 
$\mathcal{F}_{n} = \sigma ( X_{v,W}^{(1:n)}, X_W^{(n+1)})$ and that $\E[\,W^{\lambda}_n\,]=1$.
In fact, conditioned  on $\mathcal{F}_{n}$ we have that
\begin{equation*}
O_{n+1}- O_{n} = \begin{cases}
1\,, & \text{if } x_v^{(n+1)}= a_v, \,x_W^{(n+1)}= a_W\,; \\ 
0\,, & \text{c.c } \,.
\end{cases} 
\end{equation*}
and similarly 
\begin{equation*}
N_{n+1}- N_{n} = \begin{cases}
1\,, & \text{if } x_W^{(n+1)}= a_W\,; \\ 
0\,, & \text{c.c } \,.
\end{cases} 
\end{equation*}
Observe that if $x_W^{(n+1)} = a_W$  then 
\begin{equation}\label{exp_martingale}
\begin{split}
\E\left[\,2^{\lambda(O_{n+1}-O_n)}\,|\, \mathcal{F}_n\,\right]&= \E\left[\,2^{\lambda \mathbf{1}\{x_v^{(n+1)}= a_v\}}\,|\, \mathcal{F}_n\,\right]\\
& = 2^{\phi(\lambda)}\\
& = 2^{(N_{n+1}-N_n) \phi(\lambda)}\,.
\end{split}
\end{equation}
On the other hand, if  $x_W^{(n+1)} \neq a_W$
the equality trivially holds. Then  
rearranging the terms in \eqref{exp_martingale} we conclude that
\[
\E\left[ 2^{\lambda O_{n+1} - N_{n+1} \phi(\lambda)}\,|\,\mathcal{F}_n \right] = 2^{\lambda O_{n} - N_n\phi(\lambda)}
\]
and $W^{\lambda}_n$ is a martingale with respect to $\mathcal{F}_n$.
Now divide the
 interval $\{1,\dots, n\}$ of possible values of $N_n$ into ``slices'' 
 $\{t_{k-1}+1,\dots,t_k\}$ of geometrically increasing size, and treat the slices independently. We take $\alpha = \delta \log n$ and we assume that $n$ is sufficiently large so that $\alpha>1$. Take $\eta = 1/(\alpha-1)$, $t_0=0$ and for $k\geq 1$,  $t_k = \left\lfloor (1+\eta)^k \right\rfloor$.
Let $m$ be the first integer such that $t_m\geq n$, that is 
\[
m=\left\lceil\frac{\log n}{\log (1+\eta)}\right\rceil\,.
\]
Define the events  $B_k = \left\{t_{k-1} < N_n \leq t_k\right\}\cap \bigl\{ N_n\,
 | \hat{p}_n -  p|^2 > \alpha\bigr\}$. We have
\begin{equation}
\label{eq:peeling:basic}
\P\left( N_n\, | \hat{p}_n -  p|^2    \;>\; \alpha\right)  \;\leq\; \P\biggl( \bigcup_{k=1}^m B_k\biggr) \;\leq\; \sum_{k=1}^m \P\left(B_k\right).
\end{equation}
Without loss of generality we can assume that $\hat p \geq p$ (the case $\hat p \leq p$ holds by symmetry).
Observe that $| x - p|^2$ is a continuous increasing function for $x\in [p;1]$, with 
$0\leq |x-p|^2 \leq |1-p|^2$. 
Let $x$ be such that $|x-p|^2 = \alpha /(1+\eta)^k$, that is  we take
\[
x \;=\; \sqrt{\frac{\alpha}{(1+\eta)^k}} + p\,.
\]
Observe that  $x\in [p, 1]$ unless $\alpha /(1+\eta)^k  > |1-p|^2$. But in this case we have that
if $N_n \leq (1+\eta)^k$ then  
\[
\alpha \;>\; (1+\eta) ^k |1-p|^2  \;\geq\; N_n |\hat p_n - p|^2 
\]
so $\P(B_k) = 0$. So we may assume that such an $x$ always 
exists over the non-empty events $B_k$.  Moreover, on $B_k$ we have that 
$|\hat p_n - p|^2 \geq \alpha/N_n \geq \alpha/(1+\eta)^k$ then we must have $\hat p_n \geq x$.
Now take 
$\lambda =  \log(x(1-p))- \log(p(1-x))$. 
It can be verified that  $\lambda x-\phi(\lambda) = d(x;p) \geq |x-p|^2$.
Then  on $B_k$ we have that 
\[
\lambda \hat p_n - \phi(\lambda)\;\geq\; 
\lambda x -\phi(\lambda) \;\geq \; |x-p|^2\;=\;  \frac{\alpha}{(1+\eta)^k}
\;\geq\;  \frac{\alpha}{ (1+\eta)N_n}
\]
therefore 
\begin{align*}
B_k&\;\subset\; \left\{ \lambda \hat p_n - \phi(\lambda)\, >\,  \frac{\alpha}{ (1+\eta)N_n}\right\}\\
&\; \subset\; \left\{  W_n^\lambda \, >\, 2^{\alpha/(1+\eta)}\right\}.
\end{align*}
As 
$\E\bigl[W_n^\lambda \bigr] =1$, 
Markov's  inequality implies that 
\begin{align}
\P\left(B_k \right) &\;\leq\; 
\P\left(W_n^\lambda  > 2^{\alpha/(1+\eta)} \right) \\
&\;\leq \;2^{-\alpha/(1+\eta)}.\nonumber
\end{align}
Finally, by  \eqref{eq:peeling:basic}  we have that 
\[
\P\left( N_n |\hat p_n - p|^2\;>\; \alpha \right) \;\leq\; m\, 2^{-\alpha/(1+\eta)}.
 \]
But as $\eta = 1/(\alpha -1)$, $m=\left\lceil\frac{\log n}{\log(1+\eta)}\right\rceil$ and $\log(1+1/(\alpha-1))\geq 1/\alpha$ we obtain
\[
\P\left( N_n \,|\hat p_n - p|^2\;>\; \alpha \right) \; \leq \; 2 \alpha\log(n) 2 ^{-\alpha }\;=\; \frac{2\delta \log^2(n)}{n^\delta}.
\] 
Finally, by \eqref{first_eq} we obtain that 
\[
\P\Bigl( N(a_W) \,\sup_{a_v\in A} | \hat{p}(a_v |a_W) -  p(a_v |a_W)|^2 \;>\; \delta \log n \Bigr)\;\leq\; \frac{2|A|\delta \log^2 n}{n^\delta}\,.\qedhere
\]
 \end{proof}

Now we state a result controlling the probability in Proposition~\ref{prop:bound} for all  possible neighbourhoods at the same time. 

\begin{proposition}\label{prop:bound2}
For all $\delta>0$, all $v\in V_n$  and $|V_n| = o(\log n)$ we have 
\[
\P\Bigl(\,\sup_{W\subset V_n\setminus\{v\}} \sup_{a_W\in A^W}\sup_{a_v\in A} \; N(a_W) \, | \hat{p}(a_v |a_W) -  p(a_v |a_W)|^2 \;<\, \delta \log n\,\Bigr) \;\rightarrow \;1
\]
when  $n\to\infty$.  Moreover, if $\delta>1$ then the probability equals one for all sufficiently large $n$. 
\end{proposition}

\begin{proof}
Assume $|V_n| = \epsilon_n\log n$, with $\epsilon_n\to 0$ when $n\to\infty$. By Proposition~\ref{prop:bound} and a union bound we have that 
\begin{equation*}
\begin{split}
\P\Bigl(\sup_{W\subset V_n\setminus\{v\}} \sup_{a_W\in A^W}\sup_{a_v\in A} \; N(a_W) \, | \hat{p}(a_v |a_W) -  &p(a_v |a_W)|^2 \;>\; \delta\log n \Bigr)\\
&\;\leq\;\, 2^{|V_n|}|A|^{|V_n|} \frac{2|A|\delta \log^2 n}{n^\delta}\\
&\;\leq\;\, \frac{c \delta \log^2 n}{n^{\delta-2\epsilon_n}}\,.
\end{split}
\end{equation*}
For $\delta >0$ the bound on the right-hand side converges to 0 and 
it  is summable in $n$ for any $\delta>1$. Then the almost sure convergence follows by the Borel-Cantelli lemma. 
\end{proof}

The following basic result about the K\"ullback-Leibler divergence corresponds to
 \cite[Lemma~6.3]{csiszar2006b}.  We omit its proof here. 

\begin{lemma}\label{KLTV}
For any $P$ and $Q$ we have
\[
D(P;Q)  \;\leq\; \sum_{a\in A\colon Q(a)>0} \frac{[P(a)-Q(a)]^2}{Q(a)}\,.
\]
\end{lemma}

The next lemma was proved in  \cite[Lemma A.2]{csiszar2006b} for translation invariant Markov random fields. As our setting is different, we include its proof here. 

\begin{lemma} \label{vizinhanca_contida1}
If a neighborhood $W$ of $v\in V$ satisfies
\[
p(a_v|a_W)=p(a_v | a_{\nei(v)} ) 
\]
for  all $a_v \in A$, and all $a_{W\cup\nei(v)}\in A^{W\cup\nei(v)}$ with $p(a_{W\cup\nei(v)})>0$ then $W$ is a Markov neighborhood.
\end{lemma}

\begin{proof}
We have to show that for any $\Delta\subset V$ finite,  with $\Delta \supset W$,
\begin{equation}\label{igualdade2}
  p(a_v|a_\Delta) = p(a_v | a_W)
\end{equation}
for all $a_v \in A$ and all $a_{\Delta}\in A^{\Delta}$ with $p(a_{\Delta})>0$. 
As $\nei(v)$ is a Markov neighborhood, the lemma's condition implies
\[
p(a_v|a_W)= p(a_v|a_{\nei(v)})= p(a_v|a_{\nei(v) \cup \Delta})
\]
or all $a_v \in A$ and all $a_{\nei(v) \cup \Delta}\in A^{\nei(v) \cup \Delta}$ with $p(a_{\nei(v) \cup \Delta})>0$.
So (\ref{igualdade2}) follows,  because $W\subseteq \Delta \subseteq \nei(v) \cup \Delta$.
\end{proof}

The following proposition guarantees uniform control of all empirical marginal probabilities for subsets of variables in $V_n$. 

\begin{proposition}\label{joint_typ}
Let $\{V_n\}_{n\in \N}$ be such that $|V_n| = o(\log n)$.
%
Then for all $\delta>2$ we have 
\[
| \hat{p}(a_W) - p(a_W) |\; < \; \sqrt{\frac{\delta \log n}{n}}
\]
simultaneously for all $W \subseteq V_n$ and $a_W \in A^{W}$, eventually almost surely as $n \rightarrow\infty$.
\end{proposition}

\begin{proof}
For $W \subset V_n$ and $a_W \in A^{W}$ define 
\[
Y_i(a_W) = \mathbf{1}\{ x_{W}^{(i)}= a_W\} - p(a_W)\,, \quad\; i=1,2,\ldots,n.
\]
Note that $ \E (Y_i(a_W)) =0 $ and $|Y_i(a_W)|\leq 1$ for all $i=1,2,\ldots,n$. Then by Hoeffding's Inequality we have that 
\begin{align*}
\P \Bigl( \, \Bigl| \frac{1}{n} \sum_{i=1}^n Y_i(a_W) - & \E\Bigl[ \frac{1}{n} \sum_{i=1}^n Y_i(a_W) \Bigr]  \Bigr|  \geq t \Bigr)  \;\leq\; 2 \exp \bigl( - \frac{n t^2}{2} \bigr)\,. 
\end{align*}
Observe that
\[
\frac{1}{n}\sum_{i=1}^n Y_i(a_W) = \frac{N(a_W)}{n} -p(a_W) 
\]
and 
\[
\E\Bigl[ \frac{1}{n} \sum_{i=1}^n Y_i(a_W) \Bigr]  = 0\,.
\]
Therefore
\[
 \P \Bigl( \Bigl| \frac{N(a_W)}{n} - p(a_W)  \Bigr| \geq t \Bigr)\; \leq\; 2 \exp\bigl(- \frac{nt^2}{2} \bigr)\,.
 \] 
Taking $t = \sqrt{\frac{\delta \log n}{n}}$ we have that
\begin{align*}
\P \Bigl( \bigl|  \hat{p}(a_W) - &p(a_W) \bigl| \geq  \sqrt{\frac{\delta \log n}{n}} \mbox{ for some } W \subset V_n \mbox{ and } a_W \in A^W  \Bigr) \\
&\leq \;\sum_{W \subset V_n} \sum_{a_W \in A^W} \P \Bigl( \bigl| \hat{p}(a_W) - p(a_W) \bigr| \geq  \sqrt{\frac{\delta \log n}{n}} \Bigr)  \\ 
&\leq \;2^{|V _n|} |A|^{|V _n|}  \,2 \exp\bigl(- \frac{\delta \log n}{2} \bigr)  
\end{align*}
which is summable in $n$ for $\delta > 2$. This completes the proof.
\end{proof}


\begin{proof}[Proof of Theorem~\ref{main}]
Denote by
\[
\pml(x_v^{(1:n)}| x_{W}^{(1:n)})  \;=\;\log \hat\P(x_v^{(1:n)}| x_{W}^{(1:n)})  - c|A|^{|W|}\log n\,,
\]
where $\hat\P(x_v^{(1:n)}| x_{W}^{(1:n)})$ is given by \eqref{hatp}. 
If $V_n \setminus \{v\}$ is the bounding set for the candidate neighborhoods of vertex $v$ and $\nei (v)$ is the basic neighborhood of $v$, we need 
to prove that for any $c>0$ 
\begin{equation}\label{mainineq}
\max_{W \subset V_n \setminus \{v\}, W \neq \nei(v)} \pml(x_v^{(1:n)}| x_{W}^{(1:n)}) \; <\;  \pml(x_v^{(1:n)}| x_{\nei(v)}^{(1:n)})
\end{equation}
with probability converging to 1 when $n\to\infty$. Moreover, for $c > [p_{\min}(v)(|A|-1)]^{-1}$ we need to prove that \eqref{mainineq} holds eventually almost surely as $n\to\infty$.
We divide the proof in two cases, 
showing that
\begin{equation}\label{basicineq}
\max_{ W \in \mathcal{B}_i} \pml(x_v^{(1:n)}| x_{W}^{(1:n)}) \; <\;  \pml(x_v^{(1:n)}| x_{\nei(v)}^{(1:n)})
\end{equation}
with high probability or almost surely 
as $n\to\infty$, depending on the value of $c$,  for $i=1,2$ where
\begin{itemize}
\item[(a)] $\mathcal{B}_1 = \{W \subset V_n \setminus \{v\} \colon \nei(v) \subsetneq W \}$
\item[(b)] $\mathcal{B}_2 = \{W \subset V_n \setminus \{v\} \colon \nei(v) \not\subset W\}$
\end{itemize} 
For case (a),  observe that for all $W \in \mathcal{B}_1$
\begin{align}\label{eq:pml}
\pml(x_v^{(1:n)}&| x_{\nei(v)}^{(1:n)}) - \pml(x_v^{(1:n)}| x_{W}^{(1:n)})
\; = \\
&\qquad c\,(|A|^{|W|} - |A|^{|\nei (v)|})\log n \; - \sum_{a_v,a_W\in A^{|W|+1}} N(a_v,a_W) \log \frac{\hat p(a_v|a_W)}{\hat p(a_v|a_{\nei (v)})}\,.\notag
\end{align}
As these empirical probabilities are the maximum likelihood estimators of the conditional probabilities and $\nei(v)\subset W$  we have that
\begin{align*}
  \sum_{a_v,a_W\in A^{W+1}} N(a_v,a_W) \log \hat p(a_v|a_{\nei (v)}) \;&\geq\; \sum_{a_v,a_W\in A^{W+1}} N(a_v,a_W)\log p(a_v|a_{\nei (v)})\\
 &= \;  \sum_{a_v,a_W\in A^{W+1}} N(a_v,a_W) \log p(a_v|a_W)\,.
\end{align*}
Therefore, \eqref{eq:pml} can be lower-bounded by
\begin{equation}\label{diff}
c\, \left( 1 - \frac{1}{|A|}\right)  |A|^{|W|}\log n\; -\sum_{a_v,a_W\in A^{W+1}} N(a_v,a_W) \log \frac{\hat p(a_v|a_W)}{p(a_v|a_W)}\,.
\end{equation}
Note that  
\begin{align*}
\sum_{a_v,a_W\in A^{W+1}} N(a_v,a_W) \log \frac{\hat p(a_v|a_W)}{p(a_v|a_W)} 
 =  \sum_{ a_W\in A^{W}} N(a_W)D(\hat p(\cdot_v|a_W) \,;\,p(\cdot_v|a_W) )\,,
\end{align*}
where $D$ denotes the K\"ullback-Leibler divergence, see \eqref{KL}. Therefore we have, by Lemma~\ref{KLTV}, that 
\begin{equation}\label{eqD}
\begin{split}
\sum_{a_W\in A^{W}} N(a_W)&D(\hat p(\cdot_v|a_W) \,;\,p(\cdot_v|a_W) ) \\
&\leq\; \sum_{a_W\in A^{W}} N(a_W) \sum_{a_v\in A} \frac{[\,\hat p(a_v|a_W) - p(a_v|a_W)\,]^2}{p(a_v|a_W)}.
\end{split}
\end{equation}
Then, by Proposition~\ref{prop:bound2} with $\delta>0$ and \eqref{eqD} we have, with probability converging to 1 that 
\begin{align*}
\sup_{W \in \mathcal{B}_1}\sum_{a_W\in A^{W}} N(a_W)&D(\hat p(\cdot_v|a_W) \,;\,p(\cdot_v|a_W) ) \\
&\leq\;   \frac{\delta |A|^{|W|+1}  \log n }{p_*}
\end{align*}
and this holds eventually almost surely if $\delta>1$. 
Then the difference \eqref{diff} can be lower bounded by 
\begin{equation*}
c\, \left( 1 - \frac{1}{|A|}\right)  |A|^{|W|}\log n -  \frac{\delta |A|^{|W|+1}  \log n }{p_*}\;>\; 0
\end{equation*}
if $\delta < c (|A|-1)|A|^{-2}p_*$. Then, for any $c>0$ there exists a sufficiently small $\delta>0$
such that 
\[
\max_{W \in \mathcal{B}_1}  \pml(x_v^{(1:n)}| x_{W}^{(1:n)}) \; <\;  \pml(x_v^{(1:n)}| x_{\nei(v)}^{(1:n)})  
\]
with probability converging to one. 
Moreover, if $c > |A|^2[p_*(|A|-1)]^{-1}$ we can take $\delta>1$ 
in Proposition~\ref{prop:bound2} and we have that this inequality holds
eventually almost surely as $n\to\infty$. This completes the proof of \eqref{basicineq} for case  (a). \\
Finally, to prove \eqref{basicineq} in case (b)
we will first prove that 
\[
\max_{W \in \mathcal{B}_2}  \pml(x_v^{(1:n)}| x_{W}^{(1:n)}) \; \leq \;  \pml(x_v^{(1:n)}| x_{V_n \setminus \{v\} }^{(1:n)}) 
\]
eventually almost surely as $n\to\infty$. This inequality together with case (a) will imply  \eqref{basicineq} for $i=2$.  Note that we have
\begin{align*}\label{eq:pml_under}
\pml(x_v^{(1:n)}&| x_{V_n \setminus \{v\}}^{(1:n)}) - \pml(x_v^{(1:n)}| x_{W}^{(1:n)})
\\
=\; &\sum_{a_{V_n}\in A^{V_n}} N(a_v,a_{V_n \setminus \{v\}}) \log \frac{\hat p(a_v|a_{V_n \setminus \{v\}})}{\hat p(a_v|a_W)} - c\,(|A|^{|V_n|-1} - |A|^{|W|})\log n \\
= \; &\; n\Biggl[\;\sum_{a_{V_n}\in A^{V_n}} \frac{N(a_{V_n})}{n} \log \frac{\hat p(a_v|a_{V_n\setminus \{v\}})}{\hat p(a_v|a_W)} - c\,(|A|^{|V_n|-1} - |A|^{|W|}) \frac{\log n}{n}\Biggr] \,.
\end{align*}
Observe that for the second term into de brackets we have 
\[
 c\,(|A|^{|V_n|-1} - |A|^{|W|}) \frac{\log n}{n}\;\longrightarrow\; 0
 \]
when $n\to\infty$, because we are assuming $|V_n|=o(\log n)$.  For the first term
by summing and subtracting 
$\frac{N(a_{V_n})}n \log p(a_v|a_{W})$ 
inside the sum  we have that 
\begin{align} \label{eq1}
 \sum_{a_{V_n}} \;\,\frac{N(a_{V_n})}{n}& \log \frac{\hat p(a_v|a_{V_n\setminus \{v\}})}{\hat p(a_v|a_{W})} \nonumber \\
= &\;  \sum_{a_{V_n}} \left[\frac{N(a_{V_n})}{n} \log \frac{\hat p(a_v|a_{V_n\setminus \{v\}})}{ p(a_v|a_{W})} -  \frac{N(a_{V_n})}{n} \log \frac{\hat p(a_v|a_W)}{ p(a_v|a_{W})} \right].  
\end{align}
We divide again the expression in two parts. By one hand, by looking at the second term of the sum in \eqref{eq1} we have that 
\begin{align*}
\sum_{a_{V_n}} \frac{N(a_{V_n})}{n} \log \frac{\hat p(a_v|a_W)}{ p(a_v|a_{W})} 
&= \sum_{(a_v,a_W)\in A^{1+W}} \frac{N(a_v,a_W)}{n} \log \frac{\hat p(a_v|a_W)}{ p(a_v|a_W)} \nonumber \\
&= \sum_{(a_v,a_W)\in A^{1+W}} \frac{N(a_W)}{n} \hat p(a_v|a_W) \log \frac{\hat p(a_v|a_W)}{ p(a_v|a_W)} \nonumber \\
&= \sum_{a_W\in A^{W}} \frac{N(a_W)}{n} D (\hat p(\cdot_v|a_{W})\,;\, p(\cdot_v|a_W)). \nonumber 
\end{align*}
By Lemma~\ref{KLTV} and Proposition~\ref{prop:bound2}  we have that
\begin{equation}\label{eq2}
\begin{split}
\max_{W \in \mathcal{B}_2}\; \sum_{a_W\in A^{W}} \frac{N(a_W)}{n} D (\hat p(\cdot_v|a_{W})\,;\, p(\cdot_v|a_W))\;&\leq\;  \sum_{a_W\in A^{W}} \frac{N(a_W)}{n} \sum_{a_v\in A} \frac{[\,\hat p(a_v|a_W) - p(a_v|a_W)\,]^2}{p(a_v|a_W)}\\
&\leq\; \, \max_{W \in \mathcal{B}_2} \frac{ |A|^{|W|+1}\delta \log n}{p_*\,n}\;\longrightarrow\; 0
\end{split}
\end{equation}
eventually almost surely as $n \to \infty$, for $\delta>1$. On the other hand,
 as $\hat p(a_v|a_{V_n\setminus \{v\}})$ are
the maximum likelihood estimators of $ p(a_v|a_{V_n\setminus \{v\}})$ and 
 as $V_n$ will eventually contain $\nei(v)$, the first term in the sum \eqref{eq1} can 
 be lower-bounded by
\begin{equation} \label{eq3}
\begin{split}
 \sum_{a_{V_n}} \frac{N(a_{V_n})}{n} \log \frac{p(a_v|a_{V_n\setminus \{v\}})}{ p(a_v|a_{W})} &\;=\;  \sum_{a_v}\sum_{a_{W\cup\nei(v)}} \frac{N(a_v, a_{W\cup\nei(v)})}{n}  \log \frac{p(a_v|a_{\nei(v)})}{ p(a_v|a_{W})} \,.
 \end{split}
\end{equation}
By Proposition~\ref{joint_typ} we have that 
\[
\frac{N(a_v, a_{W\cup\nei(v)})}{n} \;\geq \; p(a_v,a_{W\cup\nei(v)}) - \sqrt{\frac{3\log n}{n}}
\]
eventually almost surely, simultaneously for all $W$ and all $(a_v, a_{W\cup\nei(v)}))$. Then, \eqref{eq3} can be lower bounded by 
\begin{equation} 
\begin{split}
\sum_{a_{W\cup\nei(v)}} p(a_{W\cup\nei(v)}) D (p(\cdot_v|a_{\nei(v)})\,;\, p(\cdot_v|a_W)) -  \sqrt{\frac{3\log n}{n}} \sum_{a_{W\cup\nei(v)}}  \log \frac{p(a_v|a_{\nei(v)})}{ p(a_v|a_{W})} 
 &\;\geq \;  \frac{\alpha_*}{2}
 \end{split}
\end{equation}
eventually almost surely  as  $n\to\infty$. 
Therefore
\[
\max_{W \in \mathcal{B}_2} \; \pml(x_v^{(1:n)}| x_{W}^{(1:n)}) \; \leq \;  \pml(x_v^{(1:n)}| x_{V_n \setminus \{v\} }^{(1:n)}) 
\]
eventually almost surely as $n\to\infty$.
As $V_n$ will  contain $\nei(v)$ for $n$ sufficiently large, then $V_n \setminus \{v\} \in \mathcal{B}_1$ for such values of $n$. Then, 
by case (a) we have 
\[
\pml(x_v^{(1:n)}| x_{V_n \setminus \{v\} }^{(1:n)})  \; \leq \; \max_{W \in \mathcal{B}_1}  \pml(x_v^{(1:n)}| x_{W}^{(1:n)}) \; < \; \pml(x_v^{(1:n)}| x_{\nei(v) }^{(1:n)})
\]
eventually almost surely as $n\to\infty$, and this finishes the proof of case (b). 
By combining the results of the two cases, we conclude that 
\[
\max_{W \subset V_n \setminus \{v\}, W \neq \nei(v)} \pml(x_v^{(1:n)}| x_{W}^{(1:n)}) \; <\;  \pml(x_v^{(1:n)}| x_{\nei(v)}^{(1:n)})
\]
and that $\widehat \nei(v)=\nei(v)$,  with probability converging to 1 for all $c>0$, or eventually almost surely as $n \to \infty$, for 
$c > |A|^2[p_*(|A|-1)]^{-1}$.
\end{proof}

\begin{proof}[Proof of Corollary~\ref{cor:main2}]
The proof of the corollary follows from Theorem \ref{main} and the fact that $V'$ is finite.
\end{proof}

\bibliographystyle{plain} 
\bibliography{references}  

\end{document}